\documentclass[12pt]{article} 
\usepackage[fleqn]{amsmath}
\usepackage{amssymb}
\usepackage{color}
\usepackage{epsfig}
\usepackage{epsf}
\usepackage{subcaption}
\usepackage{multirow}
\usepackage{ifthen}

\newboolean{doubleBlind}
\setboolean{doubleBlind}{true}

\newcommand{\set}[1]{\left\{ #1 \right\}}
\newcommand{\paren}[1]{\left( #1 \right)}
\newcommand{\sqparen}[1]{\left[ #1 \right]}
\newcommand{\abs}[1]{\left\vert #1 \right\vert}

\newcommand{\floor}[1]{\left\lfloor #1 \right\rfloor}

\newcommand{\Prob}[1]{\Pr\left( #1 \right)}
\newcommand{\Ex}[2][]{\mathbb{E}_{#1\!}\sqparen{#2}}
\newcommand{\Var}[1]{\operatorname{Var}\sqparen{#1}}
\newcommand{\Cov}[2]{\operatorname{Cov}\sqparen{#1, #2}}

\newcommand{\Unif}[2]{\operatorname{U}\!\paren{#1, #2}}

\newcommand{\Normal}[3][]{
\ifx\hfuzz#1\hfuzz 
\mathcal{N}\paren{#2, #3}
\else
\mathcal{N}\paren{#1; #2, #3}
\fi
}

\newcommand*\widebar[1]{%
  \hbox{%
    \vbox{%
      \hrule height 0.5pt 
      \kern0.25ex
      \hbox{%
        \kern-0.1em
        \ensuremath{#1}%
        \kern-0.1em
      }%
    }%
  }%
}

\newcommand{\comment}[1]{}      
\newcommand{\cardProc}{CAR\_EST\_PROC}

\newcounter{procCounter}

\newtheorem{theorem}{Theorem}
\newtheorem{corollary}[theorem]{Corollary}
\newtheorem{lemma}[theorem]{Lemma}

\newtheorem{thm}{Theorem}
\newtheorem{algorithm2}{Algorithm}

\newtheorem{procedure}[procCounter]{Procedure}

\usepackage{multirow}
\usepackage{epsf}

\begin{document}

\newcommand*{\QEDA}{\hfill\ensuremath{\blacksquare}}

\title{MTS Sketch for Accurate Estimation of Set-Expression Cardinalities from Small Samples}
\author{
Reuven Cohen~~~Liran Katzir~~~Aviv Yehezkel\\
Department of Computer Science\\
Technion\\
Haifa 32000, Israel\\
}
\date{\today}
\maketitle

\begin{abstract}

Sketch-based streaming algorithms allow efficient processing of big data. These algorithms use small fixed-size storage to store a summary (``sketch'') of the input data, and use probabilistic algorithms to estimate the desired quantity.
However, in many real-world applications it is impractical to collect and process the entire data stream; the common practice is thus to sample and process only a small part of it.
While sampling is crucial for handling massive data sets, it may reduce accuracy. In this paper we present a new framework that can accurately estimate the cardinality of any set expression between any number of streams using only a small sample of each stream. The proposed framework consists of a new sketch, called Maximal-Term with Subsample (MTS), and a family of algorithms that use this sketch. An example of a possible query that can be efficiently answered using the proposed sketch is, How many distinct tuples appear in tables $T_1$ and $T_2$, but not in $T_3$?  
The algorithms presented in this paper answer such queries accurately, processing only a small sample of the tuples in each table and using a constant amount of memory.
Such estimations are useful for the optimization of queries over very large database systems.
We show that all our algorithms are unbiased, and we analyze their asymptotic variance.

\end{abstract}


\section{Introduction} \label{sec:intro}

Consider a very long stream of elements $x_1, x_2, x_3, \ldots,$ with repetitions. 
Finding the number $n$ of distinct elements, known as ``the cardinality estimation problem,'' is a well-known problem in numerous applications.
The cardinality estimation problem can be generalized to set expressions over multiple streams, which yields many important applications. 
As an example, consider three large relational databases, $T_1, T_2$ and $T_3$, with a shared field $f$. Suppose we are interested in processing the query $f_1\cap f_2 \setminus f_3$, where $f_i$ is the stream of tuples in the field $f$ of $T_i$. The database system needs to determine the best (low-cost) plan for processing this query. 
To this end, every database system contains a query optimizer. The cost of a plan is usually defined according to its CPU and I/O overhead, and should be estimated according to the input/output cardinalities of each operator in the plan.
Thus, accurate cardinality estimation of set expressions over table fields in one scan and using fixed memory is crucial for query optimizations.

It is easy to use linear $O(n)$ space to produce an accurate solution to the cardinality problem. This can be done, for example, by comparing the value of a newly encountered element, $x_i$, to every (stored) value encountered so far. If the value of $x_i$ has not been seen before, it is stored and counted. 
However, for a wide range of application domains, the data set is very large, making linear space algorithms impractical. 

The challenge of processing large volumes of data that arrive at high speed has led the research community to develop new families of algorithms that work over continuous streams and produce accurate real-time estimations while guaranteeing: (a) low processing time per element, (b) fixed-size memory, which is sub-linear in the length of the stream, and (c) high estimation quality. The two main families are:
\begin{itemize}
\item Sub-linear space algorithms, also known as sketch-based streaming algorithms. These algorithms typically use a sketch, namely, a small fixed-size storage that stores a summary of the input data. Then, they employ a probabilistic algorithm on the sketch, which estimates the desired quantity \cite{Cosma:2011, Metwally:2008}.
\item Sub-linear time algorithms. These algorithms are allowed to see only a small portion of the input. A common practice is to use sampling and process only the sampled stream elements. For a recent surveys, see \cite{Sohler10, Rubinfeld11}.
\end{itemize}
As depicted in Figure \ref{fig:us}, the algorithms presented in this paper satisfy both, i.e., they require both sub-linear space and sub-linear time.

\begin{figure}[tbp]
\begin{center}
\epsfxsize=0.33\textwidth \epsffile{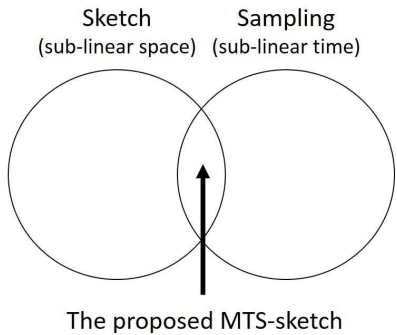}
\caption{The relationship between the two families and the proposed new MTS-sketch.}
\label{fig:us}
\end{center}
\end{figure}

Several sketch-based algorithms have been proposed for the cardinality estimation problem of both single and multiple streams \cite{Beyer07, Cosma:2011, CohenK08, Dasu02, Flajolet2007, Giroire2009, Lumbroso2010}. A common scheme, called min/max sketch, uses a hash function to map every element $x_i$ to $\Unif{0}{1}$, and then remembers only the minimum/maximum hashed value. To reduce the estimation variance, min/max sketches use $m$ different hash functions in parallel and keep the minimal/maximal hash value for each.

\ifthenelse{\boolean{doubleBlind}}
{In a previous paper \cite{Cohen15}, a sketch-based algorithm with sampling was proposed for estimating the cardinality of a single stream. The present paper extends \cite{Cohen15} for estimating the cardinality of set expressions between multiple streams.
It proposes a family of algorithms that can accurately estimate the cardinality of any set expression (intersection, union and set difference) on any number of streams, using only a small sample of each stream. We prove that all our algorithms are unbiased, and we analyze their asymptotic variance.}
{In a previous paper \cite{Cohen15}, we proposed a sketch-based algorithm with sampling for estimating the cardinality of a single stream. The present paper extends \cite{Cohen15} for estimating the cardinality of set expressions between multiple streams.
It proposes a family of algorithms that can accurately estimate the cardinality of any set expression (intersection, union and set difference) on any number of streams, using only a small sample of each stream. We prove that all our algorithms are unbiased, and we analyze their asymptotic variance. 
}

The rest of this paper is organized as follows. Section \ref{sec:related} discusses previous work. Section \ref{sec:sketch} reviews the Good-Turing frequency estimation and Jaccard similarity estimation techniques, which are used to develop our new algorithms. Section \ref{sec:algo} presents the MTS sketch and uses it to develop a family of algorithms that can accurately estimate the cardinality of any set expression between two streams using only a small sample of each. Section \ref{sec:algok} extends these algorithms to $k>2$ streams. Section \ref{sec:sim} presents simulation results and Section \ref{sec:conclusion} concludes the paper. 
A detailed analysis of the asymptotic bias and variance of the proposed algorithms is provided in the Appendix.


\section{Related Work} \label{sec:related}

Several works address the cardinality estimation problem of a single stream \cite{Cosma:2011, CohenK08, Flajolet2007, Giroire2009, Lumbroso2010, Metwally:2008} and propose sketch-based streaming algorithms for solving it. 
As already indicated, these algorithms are efficient because they make only one pass on the data stream, and because they use a fixed and small amount of storage. 

The various techniques can be classified according to the data sketch they store for future processing. For example, a min/max sketch estimator stores only the minimum/maximum hashed value. The intuition behind such estimators is that each sketch carries information about the cardinality of the stream. For example, when a hash function $h$ is used in order to associate every element $e_j$ with a uniform random variable, namely, $h(e_j) \sim \Unif{0}{1}$, the expected minimum value of $h(e_1),h(e_2), \ldots, h(e_n)$ is $1/(n+1)$. However, when only one hash function is used, the variance of the estimator is infinite. Thus, multiple different hash functions are used; the estimator keeps the minimum/maximum value for each, and then averages the results \cite{Chassaing2006, CohenK08, Giroire2009, BarYossef2002, Lumbroso2010}.

Another cardinality estimation technique is the bit pattern estimator, which keeps the highest position of the leftmost (or rightmost) ``1'' bit in the binary representation of the hash values \cite{Cosma:2011, Flajolet2007}. Bottom-$m$ sketches~\cite{CohenK08} are yet another technique. A generalization of min sketches, they maintain the $m$ minimal values, where $m \geq 1$. A comprehensive overview of the various techniques is given in~\cite{Cosma:2011, Metwally:2008}.

Sketch-based algorithms were also proposed for multiple streams \cite{Beyer07,Dasu02,Ganguly04}. An estimation of $n_{A\cup B}$, namely, the cardinality of $A\cup B$, can be found using any min/max sketch estimator for the cardinality estimation problem \cite{Gibbons16}. An estimation of $n_{A \cap B}$ can then be found using the inclusion-exclusion principle \cite{Cormode12}. In \cite{Beyer07,Dasu02,Ganguly04} it is proposed to estimate the Jaccard similarity and then use it to estimate the intersection cardinality.
In \cite{Beyer07,Ganguly04} the estimators are generalized to set expressions between more than two streams.

All the above sketch-based algorithms were designed to process the entire stream; as such, they do not use sampling. However, real-world applications that have to process large volumes of monitored data make it impractical to collect and analyze the entire input stream. Rather, the common practice is to sample and process only a small part of the stream elements. For example, routers use sampling techniques to achieve scalability \cite{sFlow1}.

Although sampling techniques provide greater scalability, they also make it more difficult to infer the characteristics of the original stream. One of the first relevant works is the Good-Turing frequency estimation, a statistical technique proposed by Alan Turing and his assistant I.J. Good, for estimating the probability of encountering a hitherto unseen element in a stream, given a set of past samples. For a recent paper on the Good-Turing technique, see \cite{Gale95}.

Many early works in the database literature tried to address the problem of estimating the cardinality from small samples; until the mid 1990s, this was the prevalent approach \cite{Gibbons16}. A sample of the data was collected and sophisticated estimators applied on the distributions of the values (see \cite{Cormode12,Gibbons16} for relevant references). 
However, because these estimators were sensitive to the order of the elements and their repetition pattern, they failed to provide accurate estimates (see pages 19-21 in \cite{Cormode12}). 

In \cite{VV11,VV13}, the authors present an estimator for the cardinality and entropy of a stream using $O(n/log{n})$ samples. Their main idea is to create a frequency histogram ``fingerprint'' of all sampled elements, and then run a linear program that approximates the real distribution in the full stream. However, creating a fingerprint requires exact mapping and counting of all the distinct elements in the given sample, whose length is $O(n/log{n})$. This becomes difficult in most real-world applications, as the number of distinct elements in the sample can be very large. The algorithm proposed in the present paper requires significantly less processing of only a small portion of the sampled stream. 

Several other works address the problem of statistical inference from samples in other computer science applications. For example, the detection of heavy hitters, namely, elements that appear many times in the stream, is studied in \cite{Bhatt07}. The authors propose to keep track of the volume of data that has not been sampled. Then, a new element is skipped only when its effect on the estimation is ``not too large." The case where the elements are network data flows has also been addressed. There, the heavy hitters (large flows) are called elephants. The accuracy of detecting elephant flows is studied in \cite{Mori07} and \cite{Mori04}. In \cite{ECohen14}, the authors study the problem of estimating the size of subpopulations of flows from a given sample. They examine some known packet sampling schemes and design unbiased sketch-based estimators for each.


\section{Preliminaries: Good-Turing Frequency Estimation and Jaccard Similarity Estimation} \label{sec:sketch}

The Good-Turing frequency estimation technique is useful in many language-related tasks where the problem is to determine the probability that a word appears in a document. Let $X=\set{x_1,x_2,x_3,\ldots}$ be a stream of elements, and $E=\set{e_1,e_2,\ldots,e_n}$ be the set of all different elements, such that $x_i \in E$. 
Suppose that we want to estimate the probability $\pi(e_j)$ that a randomly chosen element from $X$ is $e_j$. A naive approach is to choose a sample $Y = \set{y_1,y_2,\ldots,y_l}$ of $l$ elements from $X$, and then to let $\pi(e_j) = \frac{\#(e_j)}{l}$, where $\#(e_j)$ denotes the number of appearances of $e_j$ in $Y$. However, this approach is inaccurate, because for each element $e_j$ that does not appear in $Y$ even once (an ``unseen element"), $\pi(e_j)=0$. 

Let $E_i=\set{e_j | \#(e_j)=i}$ be the set of elements that appear exactly $i$ times in the sample $Y$. Thus, $\sum{\abs{E_i} \cdot i}=l$. Good-Turing frequency estimation claims that $\widehat{P_i}=(i+1)\frac{\abs{E_{i+1}}}{l}$ is a consistent estimator for the probability $P_i$ that an element of $X$ appears in the sample $i$ times.

For the case where $i=0$, we get from Good-Turing that $\widehat{P_0}=\abs{E_1}/l$. In other words, the hidden mass $P_0$ can be estimated using the relative frequency of the elements that appear exactly once in the sample $Y$. For example, if $1/10$ of the elements in $Y$ appear in $Y$ only once, then approximately $1/10$ of the elements in $X$ are unseen elements, namely, they do not appear in $Y$ at all. 

Jaccard similarity is defined as: $\rho(A,B)=\frac{\abs{A \cap B}}{\abs{A \cup B}}$, where $A$ and $B$ are two finite sets. Its value ranges between $0$, when the two sets are completely different, and $1$, when they are identical.
An efficient and accurate estimate of $\rho$ can be computed as follows \cite{Broder97new}. First, each item in $A$ and $B$ is hashed into $(0,1)$. Then, the maximal value of each set is taken as a sketch that represents the whole set. The probability that the sketches of $A$ and $B$ are equal is exactly $\rho(A,B)$ \cite{Broder97new}. When only one hash function is used, the variance of the estimate of $\rho(A,B)$ is infinite. Thus, $m$ hash functions are used\footnote{Better performance can be attained if, instead of $m$ hash functions, only two hash functions with stochastic averaging are used \cite{Flajolet:1985}.}, and the sketch representing each set is actually a vector of $m$ maximal values. 
We can state this formally as follows. Given a set $A = \set{a_1,a_2,\ldots,a_p}$ and $m$ different hash functions $h_1,h_2,\ldots,h_m$, the maximal hash value for the $j$'th hash function is $x_A^j=\max_{i=1}^{p}\set{h_j(a_i)} \text{, } 1\leq j \leq m$. Then, the sketch of $A$ is $X_A=\set{x_A^1,x_A^2,\ldots,x_A^m}$, and the sketch of $B$ is $X_B=\set{x_B^1,x_B^2,\ldots,x_B^m}$.
The two sketches can then be used to estimate the Jaccard similarity of $A$ and $B$:
\begin{equation}
\label{eq:jacest}
\widehat{\rho(A,B)}=\frac{\sum_{j=1}^{m}I_{x_A^j==x_B^j}}{m} \text{,}
\end{equation}
where the indicator variable $I_{x_A^j==x_B^j}$ is $1$ if $x_A^j=x_B^j$, and $0$ otherwise.

The Jaccard similarity can be generalized to set difference in the following way \cite{Dasu02}:
\begin{equation}\label{eq:otherJacs}
\rho(A>B)=\frac{\abs{A \setminus B}}{\abs{A \cup B}} \:\: \text{and,} \:\: \rho(A<B)=\frac{\abs{B \setminus A}}{\abs{A \cup B}} \text{.}
\end{equation}
Thus, the estimator from Eq. (\ref{eq:jacest}) can be generalized in the same way:
\begin{equation}
\label{eq:jacminus}
\widehat{\rho(A>B)}=\frac{\sum_{j=1}^{m}I_{x_A^j>x_B^j}}{m} \text{,}
\end{equation}
where the indicator function $I_{x_A^j>x_B^j}$ is $1$ if $x_A^j>x_B^j$, and $0$ otherwise. A similar estimation can be performed for $\rho(A<B)$. 

To shorten our notation, for the rest of the paper we use $\rho$, $\rho_>$ and $\rho_<$ to indicate $\rho(A,B)$, $\rho(A>B)$ and $\rho(A<B)$ respectively.


\section{MTS-based Streaming Algorithms for Set-Expression Cardinality Estimation of Two Streams} \label{sec:algo}

In this section we present the MTS sketch and use it to develop a family of algorithms that can accurately estimate the cardinality of any set expression between two streams using only a small sample of each. In Section \ref{sec:algok}, these algorithms are extended for $k>2$ streams. Table \ref{table:notations} shows some of the notations used for the rest of the paper.

\begin{table}[htb!]
\centering 
	\begin{tabular}{|c|c|} \hline
   notation & meaning   \\\hline
	$P_0^T$ & probability for unseen elements in $T$  \\\hline
$P_1^T$ & prob. for elements to appear exactly once in $T$  \\\hline
$P_{0,1}^T$ & $2P_0^T(1-P_0^T) + P_1^T$  \\\hline
	$P_0^\cup$ & probability for unseen elements in $A \cup B$  \\\hline
	$S_{T}$ & sampled stream of $T$  \\\hline
	$S_{\cup}$ & sampled stream of $A\cup B$  \\\hline
	$d_{W}$ & $\frac{\abs{S_{W}}}{\abs{S_{\cup}}}$, where $W$ is the stream $A$ or $B$ \\\hline
	$X_{T}$ & MTS maximal hash values of $T$  \\\hline
	$U_{T}$ & MTS subsample of $T$ \\\hline
	$\rho$ & $\rho(A,B)$ \\\hline
	$\rho_>$ & $\rho(A>B)$ \\\hline
	$\rho_<$ & $\rho(A<B)$ \\\hline
	\end{tabular}
	\caption{Notations ($T$ represents any stream, e.g., $A$, $B$, $A \cup B$, etc.)} 
	\label{table:notations}
\end{table}

\subsection{The MTS Sketch}

\ifthenelse{\boolean{doubleBlind}}
{In a previous paper \cite{Cohen15}, a generic scheme that combines a sampling process with a cardinality estimation procedure of a single stream was presented. The scheme in \cite{Cohen15} consists of two steps: (a) cardinality estimation of the sampled stream using any known cardinality estimator; (b) estimation of the sampling ratio, namely, the factor by which the cardinality of the sampled stream should be multiplied in order to estimate the cardinality of the full stream. The main idea was to store a small fixed-size subsample of the sampled stream and use it to estimate the probability of unseen elements using the Good-Turing technique.}
{In a previous paper \cite{Cohen15} we presented a generic scheme that combines a sampling process with a cardinality estimation procedure of a single stream. The scheme in \cite{Cohen15} consists of two steps: (a) cardinality estimation of the sampled stream using any known cardinality estimator; (b) estimation of the sampling ratio, namely, the factor by which the cardinality of the sampled stream should be multiplied in order to estimate the cardinality of the full stream. The main idea was to store a small fixed-size subsample of the sampled stream and use it to estimate the probability of unseen elements using the Good-Turing technique.}

In this paper we generalize the above mentioned scheme to set expressions between multiple streams. The main idea is to maintain, as part of the sketch of each stream, a small fixed-size subsample of the sampled stream, and use this subsample for estimating the probability of unseen elements. 
To this end, an MTS sketch stores two data structures for each sampled stream:
\begin{itemize}
\item $\text{MTS}_1$: The maximal hash value for each hash function: $H^+ = h_1^+,h_2^+,\ldots,h_m^+$.
\item $\text{MTS}_2$: A small fixed-size uniform subsample $U$ of the sample stream (see Figure \ref{fig:relationship}), used for estimating the probability of unseen elements.
\end{itemize}

Let $X=\set{x_1,x_2,x_3,\ldots,x_s}$ be a full stream of elements, and $Y=\set{y_1, y_2, \ldots, y_l}$ be a sampled stream of $X$. Assume that the sampling rate is $P$, namely, $1/P$ of the elements of $X$ are \textbf{randomly} sampled into $Y$. The subsample $U$ (Figure \ref{fig:relationship}) can be generated using one-pass reservoir sampling \cite{Vitter85}, as follows. First, $U$ is initialized with the first $u$ elements of $Y$, namely, $y_1,y_2,\ldots,y_u$, and the elements are then sorted in decreasing order of their hash values. When a new element is sampled into $Y$, its hash value is compared to the current maximal hash value of the elements in $U$. If the hash value of the new element is smaller than the current maximal hash value of $U$, the new value is stored in $U$ instead of the element with the maximal hash value. Otherwise, the new element is ignored. After all of the elements of the sample $Y$ are treated, $U$ holds the $u$ elements whose hash values were minimum, and it can be considered as a uniform subsample of length $u$.

\begin{figure}[tbp]
\begin{center}
\epsfxsize=0.33\textwidth \epsffile{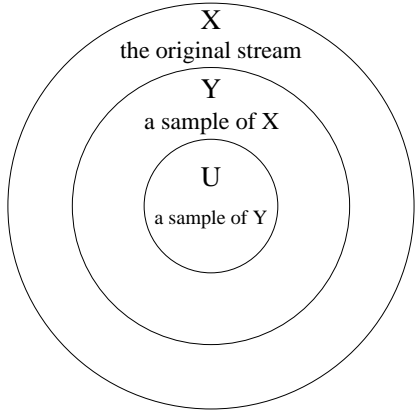}
\caption{The relationship between X, Y and U.}
\label{fig:relationship}
\end{center}
\end{figure}

Notice that the MTS sketch is additive, i.e., the MTS sketch of a union of streams can be computed directly from the MTS sketches of the streams. The next corollary summarizes this property for two streams, and it can be generalized for any $k>2$ streams as well:
\begin{corollary} \label{lemma:mtsunion} \ \\
Let $A$ and $B$ be two streams whose MTS sketches are:
\begin{itemize}
\item $\text{MTS}_1(A)$: $x_A^1,x_A^2,\ldots,x_A^m$ are the $m$ maximal hash values of $A$; $\text{MTS}_2(A) = U_A$ is a subsample of $A$ of length $u$.
\item $\text{MTS}_1(B)$: $x_B^1,x_B^2,\ldots,x_B^m$ are the $m$ maximal hash values of $B$; $\text{MTS}_2(B) = U_B$ is a subsample of $B$ of length $u$.
\end{itemize}
Then, the MTS sketch of $A \cup B$ is:
\begin{itemize}
\item[] $\text{MTS}_1(A\cup B) = H^+_{A\cup B} = h_1^+,h_2^+,h_m^+$, where $h_i^+ = \max{\set{x_A^i,x_B^i}}$.
\item[] $\text{MTS}_2(A\cup B) = U_{A \cup B}$ is the $u$ minimal hash values in $U_A \cup U_B$.
\end{itemize}
\end{corollary}

\subsection{Cardinality Estimation of a Single Stream with Sampling}

In \cite{Cohen15}, Good-Turing is used to combine a sampling process with a generic cardinality estimation procedure of a single stream. 
The algorithm receives the sampled stream as an input and returns an estimate for $n$. The algorithm consists of two steps: (a) estimating $n_s$ using \cardProc{}\footnote{Any procedure for estimating the cardinality of a single stream without sampling, such as in \cite{Cosma:2011, Flajolet2007, Giroire2009, Lumbroso2010}. This procedure is called Procedure 1 in \cite{Cohen15}.}; (b) estimating $n/{n_s}$, the factor by which the cardinality $n_s$ of the sampled stream should be multiplied in order to estimate the cardinality $n$ of the full stream.

To estimate $n_s$ in step (a), \cardProc{} is invoked using $m$ storage units. To estimate $n/{n_s}$ in step (b), we note that $P_0 = (n-n_s)/n$ and thus $1/(1-P_0) = n/{n_s}$. Therefore, the problem of estimating $n/{n_s}$ is reduced to estimating the probability $P_0$ of unseen elements.
According to Good-Turing, $\widehat{P_0}=\abs{E_1}/l$ is a consistent estimator for $P_0$, as described in Section \ref{sec:sketch}. Thus, we only need to find the number $\abs{E_1}$ of elements that appear exactly once in the sampled stream. To compute the value of $\abs{E_1}$ precisely, one should keep track of all the elements in the sample and ignore each previously encountered element. To this end, $O(l)$ storage units are needed, which is linear in the sample size and is not scalable. We reduce this cost by approximating the value of $\abs{E_1}/l$ using the subsample $\text{MTS}_2$ of the sampled stream. The above algorithm can be formulated as follows.

\begin{algorithm2}
\textbf{\newline (cardinality estimation of a single stream $X$ using the proposed MTS)}
\label{alg:cardEstWithSampling}
\begin{enumerate}
\item [(a)] Estimate the number $n_s$ of distinct elements in the sample $Y$ by invoking \cardProc{} on $\text{MTS}_1 = H^+$; namely, $\widehat{n_s} = \cardProc{}(H^+)$.
\item [(b)] Estimate the ratio $n/{n_s}$ by computing $\frac{1}{1-\widehat{P_0}}$, where $\widehat{P_0} = \abs{E_1}/l$. The value of 
$\abs{E_1}/l$ is estimated by invoking Procedure \ref{proc:p0} on the subsample $U$.
\item [(c)] Return $\widehat{n} = \widehat{n_s} \cdot \widehat{n/{n_s}}$ as an estimator for the cardinality of the entire stream $X$.
\end{enumerate}
\end{algorithm2}

\begin{procedure}
\textbf{\newline (estimation of $P_0$ from $\text{MTS}_2$)}
\label{proc:p0}
\begin{enumerate}
\item Compute (precisely) the number $\abs{U_1}$ of elements that appear only once in $\text{MTS}_2 = U$.
\item Return $\widehat{P_0} = \abs{U_1}/u$.
\end{enumerate}
\end{procedure}

\subsection{Cardinality Estimation of Set Union with Sampling}\label{sub:union}

Let $S_A$ and $S_B$ be the samples of $A$ and $B$ respectively. Let $S_AS_B$ be the concatenation of these samples. $S_AS_B$ is actually a sample of $A \cup B$, i.e., $S_AS_B = S_\cup$ (recall our notations from Table \ref{table:notations}). Thus, estimating the cardinality of $A\cup B$ is equivalent to estimating the cardinality of a single stream $A\cup B$ using $S_{\cup}$. To this end, we present Algorithm \ref{alg:unionEstWithSampling}, which uses Algorithm \ref{alg:cardEstWithSampling} on the MTS sketch of $S_\cup$.

\begin{algorithm2}
\textbf{\newline (estimating the cardinality of $A\cup B$ using $S_A$ and $S_B$)}
\label{alg:unionEstWithSampling}
\begin{enumerate}
\item [(a)] Maintain the MTS sketch for $A$ and $B$ as follows:
\begin{itemize}
\item $\text{MTS}_1(A) = x_{A}^1,x_{A}^2,\ldots,x_{A}^m$ are the $m$ maximal hash values of $A$; $\text{MTS}_2(A) = U_{A}$ is a subsample of $A$ of length $u$.
\item $\text{MTS}_1(B) = x_{B}^1,x_{B}^2,\ldots,x_{B}^m$ are the $m$ maximal hash values of $B$; $\text{MTS}_2(B) = U_{B}$ is a subsample of $B$ of length $u$.
\end{itemize}
\item [(b)] Compute the MTS sketch of $A\cup B$ according to Corollary \ref{lemma:mtsunion} as follows:
\begin{itemize}
\item[] $\text{MTS}_1(A\cup B) = H^+ = h_1^+,h_2^+,\ldots,h_m^+$, where $h_i^+ = \max{\set{x_{A}^i,x_{B}^i}}$.
\item[] $\text{MTS}_2(A\cup B) = U$ is the $u$ minimal hash values in $U_{A} \cup U_{B}$.
\end{itemize}
\item [(c)] Estimate $\widehat{\abs{A\cup B}}$ using Algorithm \ref{alg:cardEstWithSampling} and $\text{MTS}(A\cup B)$.
\end{enumerate}
\end{algorithm2}

\subsection{Cardinality Estimation of Set Intersection with Sampling}

As observed in \cite{Dasu02}, $\abs{A\cap B}= \abs{A\cup B} \cdot \rho$, where $\rho$ is the Jaccard similarity of the two full streams $A$ and $B$. 
Algorithm \ref{alg:unionEstWithSampling} can be used for estimating $\abs{A\cup B}$. To estimate the Jaccard similarity $\rho$, we note that
\begin{align}\label{eq:j1}
\rho + 1 &= \frac{\abs{A \cap B}}{\abs{A \cup B}} + 1 = \frac{\abs{A \cap B}+\abs{A \cup B}}{\abs{A \cup B}} \nonumber \\
&= \frac{\abs{A} + \abs{B}}{\abs{A \cup B}} 
= \frac{\abs{A} + \abs{B}}{\abs{A \cup B}} \cdot \frac{\abs{S_{\cup }}}{\abs{S_{\cup}}} \nonumber \\
&= \frac{\abs{A} + \abs{B}}{\abs{S_{\cup}}} \cdot \frac{\abs{S_{\cup}}}{\abs{A \cup B}} 
\text{.}
\end{align}
Recall that according to Good-Turing, it follows that (see the notations in Table \ref{table:notations})
\begin{equation}\label{eq:j2}
\abs{A} = \frac{1}{1-P_0^A}\cdot \abs{S_A} \text{,} \:\: \text{or equivalently} \:\: \frac{\abs{S_A}}{\abs{A}} = 1-P_0^A {.}
\end{equation}
Similar equations can be written for $\frac{\abs{S_B}}{\abs{B}}$ and for $\frac{\abs{S_\cup}}{\abs{A\cup B}}$. Substituting Eq. (\ref{eq:j2}) into Eq. (\ref{eq:j1}) yields that
\begin{align*}
\rho + 1 &= \big(\frac{1}{1-P_0^A}\cdot\frac{\abs{S_A}}{\abs{S_{\cup}}} + \frac{1}{1-P_0^B}\cdot\frac{\abs{S_B}}{\abs{S_{\cup}}}\big)\cdot \frac{\abs{S_{\cup}}}{\abs{A \cup B}} \\
&= \big(\frac{1}{1-P_0^A}\cdot\frac{\abs{S_A}}{\abs{S_{\cup}}} + \frac{1}{1-P_0^B}\cdot\frac{\abs{S_B}}{\abs{S_{\cup}}}\big)\cdot (1-P_0^\cup) \text{,}
\end{align*}
or equivalently
\begin{equation}\label{eq:jacInter2}
\widehat{\rho} = \big(\widehat{\frac{1}{1-P_0^A}}\cdot\widehat{\frac{\abs{S_A}}{\abs{S_{\cup}}}} + \widehat{\frac{1}{1-P_0^B}}\cdot\widehat{\frac{\abs{S_B}}{\abs{S_{\cup}}}}\big)\cdot \widehat{(1-P_0^\cup)} - 1\text{.}
\end{equation}
Denoting $d_A=\frac{\abs{S_A}}{\abs{S_\cup}}$ and $d_B=\frac{\abs{S_B}}{\abs{S_\cup}}$, Eq. (\ref{eq:jacInter2}) can be rewritten as follows:
\begin{equation}\label{eq:jacEst2}
\widehat{\rho}= \big(\widehat{\frac{1}{1-P_0^A}}\cdot\widehat{d_A} + \widehat{\frac{1}{1-P_0^B}}\cdot\widehat{d_B}\big)\cdot \widehat{(1-P_0^\cup)} - 1
\text{.}
\end{equation}
We now present Algorithm \ref{alg:interEstWithSampling}, for estimating $\abs{A \cap B}$. In this algorithm, $P_0^A$ and $P_0^B$ are estimated using Procedure \ref{proc:p0}. $P_0^\cup$ can also be estimated using Procedure \ref{proc:p0}, from $\text{MTS}_2(A\cup B)$. Finally, $d_A$ and $d_B$ are estimated using Procedure \ref{proc:da} (presented after Algorithm \ref{alg:interEstWithSampling}).

\begin{algorithm2}
\textbf{\newline (estimating the cardinality of $A \cap B$ using $S_A$ and $S_B$)}
\label{alg:interEstWithSampling}
\begin{enumerate}
\item [(a)] Maintain the MTS sketch for $A$ and $B$ as follows:
\begin{itemize}
\item $\text{MTS}_1(A) = x_{A}^1,x_{A}^2,\ldots,x_{A}^m$ are the $m$ maximal hash values of $A$; $\text{MTS}_2(A) = U_{A}$ is a subsample of $A$ of length $u$.
\item $\text{MTS}_1(B) = x_{B}^1,x_{B}^2,\ldots,x_{B}^m$ are the $m$ maximal hash values of $B$; $\text{MTS}_2(B) = U_{B}$ is a subsample of $B$ of length $u$.
\end{itemize}
\item [(b)] Use Procedure \ref{proc:p0} to estimate $P_0^A$,$P_0^B$ and $P_0^\cup$.
\item [(c)] Use Procedure \ref{proc:da} to estimate $d_A$ and $d_B$.
\item [(d)] Use Algorithm \ref{alg:unionEstWithSampling} to estimate $\abs{A\cup B}$.
\item [(e)] Use the estimations for $P_0^A$,$P_0^B$,$P_0^\cup$,$d_A$ and $d_B$ to estimate the Jaccard similarity $\rho$ according to Eq. (\ref{eq:jacEst2}).
\item [(f)] Return $\widehat{\abs{A\cap B}} = \widehat{\abs{A\cup B}} \cdot \widehat{\rho}$.
\end{enumerate}
\end{algorithm2}

\setcounter{procCounter}{2}

\begin{procedure}
\textbf{\newline (estimation of $d_A$ from $\text{MTS}_2(A)$ and $\text{MTS}_2(B)$)}
\label{proc:da}
\begin{enumerate}
\item Compute $\text{MTS}_2(A\cup B)$ according to Corollary \ref{lemma:mtsunion}. 
\item Compute (precisely) the cardinalities of $\abs{\text{MTS}_2(A\cup B)}$ and $\abs{\text{MTS}_2(A) \cap \text{MTS}_2(A\cup B)}$.
\item Return $\widehat{d_A} = \frac{\abs{\text{MTS}_2(A) \cap \text{MTS}_2(A\cup B)}}{\abs{\text{MTS}_2(A\cup B)}}$.
\end{enumerate}
\end{procedure}

\subsection{Cardinality Estimation of Set Difference with Sampling}

As observed in \cite{Dasu02}, $\abs{A \setminus B}= \abs{A \cup B} \cdot \rho_>$, where $\rho_>=\rho(A>B)$ (see Eq. (\ref{eq:otherJacs})).
Thus, Algorithm \ref{alg:interEstWithSampling} can be used for estimating $\abs{A\setminus B}$, with the only difference being that $\rho_>$ rather than $\rho$ has to be estimated. We note that
\begin{align}\label{eq:j3}
1 - \rho_> &= 1 - \frac{\abs{A \setminus B}}{\abs{A \cup B}} = \frac{\abs{A\cup B} - (\abs{A\cup B} - \abs{B})}{\abs{A\cup B}} \nonumber \\
&= \frac{\abs{B}}{\abs{A \cup B}} = \frac{\abs{B}}{\abs{A \cup B}} \cdot \frac{\abs{S_{\cup}}}{\abs{S_{\cup}}} \nonumber \\
&= \frac{\abs{B}}{\abs{S_{\cup}}}\cdot \frac{\abs{S_{\cup}}}{\abs{A\cup B}}
\text{.}
\end{align}
Eq. (\ref{eq:j3}) follows from the inclusion-exclusion principle and some elementary algebraic manipulations. By substituting Eq. (\ref{eq:j2}) into Eq. (\ref{eq:j3}), we get
\begin{equation}\label{eq:jacDiff2}
\widehat{\rho_>} = 1- \widehat{\frac{1}{1-P_0^B}}\cdot\widehat{\frac{\abs{S_B}}{\abs{S_{\cup}}}}\cdot \widehat{(1-P_0^\cup)} \text{.}
\end{equation}
Recall our notation from Table \ref{table:notations}, where $d_B=\frac{\abs{S_B}}{\abs{S_{\cup}}}$. We can rewrite Eq. (\ref{eq:jacDiff2}) as follows:
\begin{equation}\label{eq:jacDiffEst2}
\widehat{\rho(A>B)}= 1 - \widehat{\frac{1}{1-P_0^B}}\cdot\widehat{d_B}\cdot \widehat{(1-P_0^\cup)}
\text{.}
\end{equation}
We now present Algorithm \ref{alg:diffEstWithSampling} for estimating $\abs{A \setminus B}$. In this algorithm, $P_0^B$ and $P_0^\cup$ are estimated using Procedure \ref{proc:p0}. In addition, $d_B$ is estimated using Procedure \ref{proc:da}.

\begin{algorithm2}
\textbf{\newline (estimating the cardinality of $\abs{A \setminus B}$ using $S_A$ and $S_B$)}\ \\
\label{alg:diffEstWithSampling}
Same as in Algorithm \ref{alg:interEstWithSampling} except that $\rho_>$ rather than $\rho$ is estimated (using Eq. (\ref{eq:jacDiffEst2})).
\end{algorithm2}


\section{MTS-based Streaming Algorithms for Estimating the Cardinality of Any Set Expression between $k>2$ Streams}\label{sec:algok}

We now present a generic MTS-based algorithm for estimating the cardinality of any set expression between $k>2$ streams. 
Let $A_1,A_2,\ldots,A_k$ be the original streams, and $S_1,S_2,\ldots,S_k$, be their sampled streams. Our goal is to use these sampled streams in order to estimate the cardinality of $X=A_1 (\star 1) A_2 (\star 2) \ldots (\star (k-1)) A_k$, where $(\star t) \in {\cup,\cap,\setminus}$.

Let $Y = S_1 (\star 1) S_2 (\star 2) \ldots (\star (k-1)) S_k$. We denote the cardinalities of $X$ and $Y$ by $n$ and $n_s$, respectively. Let $\rho_G$ be the ``generalized'' Jaccard similarity $\rho_G = \frac{n_s}{\abs{S_\cup}}$. It can be estimated from $\text{MTS}_1(A_1),\ldots,\text{MTS}_1(A_k)$ in a similar way to the estimation of $\rho$ in Eq. (\ref{eq:jacest}), namely,
\begin{equation}\label{eq:estRG}
\widehat{\rho_G} = \frac{\sum_{j=1}^{m}{I_j}}{m}\text{,}
\end{equation}
where the indicator variable $I_j$ is $1$ if, for the $j$'th hash function, $\text{MTS}_1(A_1),\ldots,\text{MTS}_1(A_k)$ satisfy the condition implied by the set expressions, and is $0$ otherwise.

Using algebraic manipulations and the definition of $\rho_G$ we obtain that
\begin{equation*}
n = n_s\cdot\frac{n}{n_s} = \frac{n_s}{\abs{S_\cup}}\cdot \abs{S_\cup} \cdot \frac{n}{n_s} = \rho_G \cdot \abs{S_\cup} \cdot \frac{n}{n_s} \text{.}
\end{equation*}
Thus, we can estimate $n$ by:
\begin{equation*}
\widehat{n} = \widehat{\rho_G}\cdot\widehat{\abs{S_\cup}}\cdot \widehat{\frac{1}{1-P_0^X}}\text{.}
\end{equation*}

We now present Algorithm \ref{alg:generalEstWithSampling} for estimating $n$.
The algorithm consists of three steps: (a) using Eq. (\ref{eq:estRG}) to estimate $\rho_G$; (b) using \cardProc{} to estimate $\abs{S_\cup}$; and (c) using Procedure \ref{proc:gen} to estimate $\frac{n}{n_s}$, the factor by which the cardinality $n_s$ of the sampled stream should be multiplied in order to estimate the cardinality $n$ of the full stream.

\begin{algorithm2}
\textbf{\newline (estimating any set-expression cardinality of $k>2$ streams with sampling using MTS)}
\label{alg:generalEstWithSampling}
\begin{enumerate}
\item [(a)] Maintain the MTS sketch for $A_i$ as follows:
\begin{itemize}
\item[] $\text{MTS}_1(A_i) = x_{A_i}^1,x_{A_i}^2,\ldots,x_{A_i}^m$ be the $m$ maximal hash values of $A_i$; 
\item[] $\text{MTS}_2(A_i) = U_{A_i}$ is a subsample of $A_i$ of length $u$.
\item [(b)] Use Eq. (\ref{eq:estRG}) to estimate $\rho_G$.
\item [(c)] Use \cardProc{} to estimate $\abs{S_\cup}$.
\item [(d)] Use Procedure \ref{proc:gen} to estimate $\frac{1}{1-P_0^X}$.
\item [(e)] Return $\widehat{n} = \widehat{\rho_G}\cdot \widehat{\abs{S_\cup}}\cdot \widehat{\frac{1}{1-P_0^X}}$.
\end{itemize}
\end{enumerate}
\end{algorithm2}

\setcounter{procCounter}{4}

\begin{procedure}
\textbf{\newline (estimation of $\frac{1}{1-P_0^X}$ from $\text{MTS}_2(A_1),\ldots,\text{MTS}_2(A_k)$)}
\label{proc:gen}
\begin{enumerate}
\item Compute $U_X = \text{MTS}_2(A_1)(\star 1)\text{MTS}_2(A_2)(\star 2)\ldots (\star (k-1))\text{MTS}_2(A_k)$.
\item If the length of $U_X$ is greater than $u$, then keep in $U_X$ only the $u$ elements with the minimal hash values.
\item Let $f_1$ be the smallest frequency of element in $U_X$; compute (precisely) the value of $\abs{U_1}$, namely, the number of elements that appear exactly $f_1$ times in $U_X$.
\item Compute $\widehat{P_0^X} = \frac{\abs{U_1}}{f}$, where $f$ is the length of $U_X$.
\item Return $\frac{1}{1-\widehat{P_0^X}}$.
\end{enumerate}
\end{procedure}


\section{Simulation Study} \label{sec:sim}

In this section we validate our analysis from the Appendix for the asymptotic bias and variance of the presented MTS algorithms. More specifically, we show that
\begin{itemize}
\item Algorithms \ref{alg:interEstWithSampling} and \ref{alg:diffEstWithSampling} are unbiased, as proven by Theorems \ref{thmInter} and \ref{thmDiff};
\item the variance of Algorithms \ref{alg:interEstWithSampling} and \ref{alg:diffEstWithSampling} is close to their analyzed variance in Theorems \ref{thmInter} and \ref{thmDiff};
\item the variance of Algorithm \ref{alg:generalEstWithSampling} is close to its analyzed variance in Theorem \ref{thmKbigT2}.
\end{itemize} 

We implement the algorithms using the HyperLogLog \cite{Flajolet2007} as \cardProc{}, and simulate two sets, $A$ and $B$, whose cardinalities are as follows:
\begin{itemize}
\item $\abs{A}=a=10^4$;
\item $\abs{B} = a \cdot f$, where $f>0$;
\item $\abs{A \cap B} = a \cdot \alpha$, where $0\leq\alpha \leq 1$.
\end{itemize}
Each distinct element $e_j$ appears $f_j$ times in the original (unsampled) stream. The $f_j$ frequencies are determined according to the following models:
\begin{itemize}
\item Uniform distribution: The frequency of the elements is uniformly distributed between $100$ and $1,000$; i.e., $f_j \sim \Unif{10^2}{10^3}$.
\item Pareto distribution: The frequency of the elements follows the heavy-tailed rule with shape parameter $p$ and scale parameter $s=500$; i.e., the frequency probability function is $p(f_j) = p s^\alpha f^{-p-1}$, where $p>0$ and $f_j \ge s >0$. The scale parameter $s$ represents the smallest possible frequency.
\end{itemize}

Pareto distribution has several unique properties. In particular, if $\alpha \le 2$, it has infinite variance, and if $\alpha \le 1$, it has infinite mean. 
As $\alpha$ decreases, a larger portion of the probability mass is in the tail of the distribution, and it is therefore useful when a small percentage of the population controls the majority of the measured quantity.

Each of the simulation tests described below is repeated for $1,000$ different sets. Thus, for each algorithm and for each $\alpha$ value we get a vector of $1,000$ different estimations. Then, for each $\alpha$ value, we compute the variance and bias of this vector, and view the result as the variance and bias of the algorithm (for the specific $\alpha$ value). Each such computation is represented by one table row. 
Let $v_{\alpha}=(\widehat{n}_1,\ldots,\widehat{n}_{10^3})$ be the vector of estimations for a specific algorithm and for a specific $\alpha$ value. Let $\mu = \frac{1}{10^3}\sum_{i=1}^{10^3}{\widehat{n}_i}$ be the mean of $v_{\alpha}$.
The bias and variance of $v_{\alpha}$ are computed as follows:
\begin{equation*}
\text{Bias}(v_{\alpha})=\abs{\frac{1}{n}(\mu - n)}
\end{equation*}
and
\begin{equation*}
\Var{v_{\alpha}}=\frac{1}{10^3}\sum_{i=1}^{10^3}{(\widehat{n}_i-\mu)^2}.
\end{equation*}

\subsection{The Case of Two Streams: Algorithms \ref{alg:interEstWithSampling} and \ref{alg:diffEstWithSampling}}

First, we verify the unbiasedness of Algorithms \ref{alg:interEstWithSampling} and \ref{alg:diffEstWithSampling}.
Table~\ref{table:biasInterAndDiff} presents the simulation results for different $\alpha$ values using uniformly distributed frequencies ($m=10$ buckets and $u=300$) and Pareto distributed frequencies ($m=100$ buckets and $u=1,000$). The sampling ratio is $P=1/100$ and $f=1$.
In each table row we present the bias. The results in Table \ref{table:biasInterAndDiff} show very good agreement between the simulation results and our analysis, because all bias values are very close to $0$.

For the uniform distribution, the number of distinct elements is $n=10,000$. Thus, the expected length of each original stream is $ 10,000 \cdot \frac{100 + 1,000}{2} = 5.5 \cdot 10^6$.
We can see that a total budget of $B=10+300=310$ storage units per stream, which is about $0.006 \%$ of the stream length, yields accurate estimation of both intersection and difference cardinalities. For the Pareto distribution, the expected length of each original stream is $\approx 500 \cdot 10^6$. Using a total budget of $B=100+1,000=1,100$ storage units, namely, $2\cdot 10^{-6}$ of the stream length, yields very accurate estimations.

\begin{table}[ht]
\centering 
\begin{tabular}{|c|c|c|c|c|}
\hline
\multirow{2}{*}{$\alpha$} & \multicolumn{2}{|c|}{uniform} & \multicolumn{2}{|c|}{pareto} \\ \cline{2-5}
{} & Alg. 3 & Alg. 4 & Alg. 3 & Alg. 4 \\\hline 
0.1 & 0.0056 & 0.0089 & 0.0355 & 0.0143 \\\hline
0.3 & 0.0013 & 0.0171 & 0.0152 & 0.0001 \\\hline 
0.5 & 0.0114 & 0.0087 & 0.0117 & 0.0052 \\\hline 
0.7 & 0.0167 & 0.0093 & 0.0165 & 0.0294 \\\hline
0.9 & 0.0274 & 0.0379 & 0.0136 & 0.0235 \\\hline
\end{tabular}
\caption{Simulation results for the bias of Algorithms \ref{alg:interEstWithSampling} and \ref{alg:diffEstWithSampling} using uniform and Pareto distributions. All values are indeed very close to $0$.}
\label{table:biasInterAndDiff}
\end{table}

We now study the variance of Algorithms \ref{alg:interEstWithSampling} and \ref{alg:diffEstWithSampling}.
Tables~\ref{table:varInter} and \ref{table:varDiff} present the simulation results of both algorithms for different $\alpha$ values using uniform and Pareto frequency distributions. In both tables, $m=100$ buckets and $u=1,000$. The sampling ratio is $P=1/100$ and we use two values of $f$, $f=1$ and $f=3$. The results are averaged over $1,000$ runs, 
and the ``analysis'' variance is determined according to Theorems \ref{thmInter} and \ref{thmDiff}. 
We can see that there is excellent agreement between the simulation results and the results expected by our analysis: the relative error is always less than $20\%$, and mostly less than $10\%$.

\begin{table}[ht]
\centering
\begin{subfigure}[b]{0.44\textwidth}
\begin{tabular}{|c|c|c|c|c|}
\hline
\multirow{2}{*}{$\alpha$} & \multicolumn{2}{|c|}{uniform} & \multicolumn{2}{|c|}{pareto} \\ \cline{2-5}
{} & simulations & analysis & simulations & analysis \\\hline 
0.1 & 0.1228 & 0.1030 & 0.5443 & 0.4787 \\\hline
0.3 & 0.0403 & 0.0373 & 0.1495 & 0.1387 \\\hline 
0.5 & 0.0255 & 0.0236 & 0.0809 & 0.0672 \\\hline 
0.7 & 0.0174 & 0.0171 & 0.0415 & 0.0399 \\\hline
0.9 & 0.0127 & 0.0128 & 0.0193 & 0.0183 \\\hline
\end{tabular}
\caption{$f=1$}
\label{table:varInterFrac1}
\end{subfigure}
\\
\begin{subfigure}[b]{0.44\textwidth}
\begin{tabular}{|c|c|c|c|c|}
\hline
\multirow{2}{*}{$\alpha$} & \multicolumn{2}{|c|}{uniform} & \multicolumn{2}{|c|}{pareto} \\ \cline{2-5}
{} & simulations & analysis & simulations & analysis \\\hline 
0.1 & 0.1814 & 0.2058 & 0.9982 & 0.9671 \\\hline
0.3 & 0.0848 & 0.0706 & 0.3146 & 0.2766 \\\hline 
0.5 & 0.0516 & 0.0442 & 0.1972 & 0.1955 \\\hline 
0.7 & 0.0391 & 0.0327 & 0.1432 & 0.1258 \\\hline
0.9 & 0.0284 & 0.0262 & 0.0946 & 0.0929 \\\hline
\end{tabular}
\caption{$f=3$}
\label{table:varInterFrac3}
\end{subfigure}
\caption{Simulation results for the variance of Algorithm \ref{alg:interEstWithSampling} using uniform and Pareto distributions.
The ``analysis'' variance is determined according to Theorem \ref{thmInter}.}
\label{table:varInter}
\end{table}

\begin{table}[ht]
\centering
\begin{subfigure}[b]{0.44\textwidth}
\begin{tabular}{|c|c|c|c|c|}
\hline
\multirow{2}{*}{$\alpha$} & \multicolumn{2}{|c|}{uniform} & \multicolumn{2}{|c|}{pareto} \\ \cline{2-5}
{} & simulations & analysis & simulations & analysis \\\hline 
0.1 & 0.0158 & 0.0170 & 0.0407 & 0.0371 \\\hline
0.3 & 0.0202 & 0.0200 & 0.0519 & 0.0457 \\\hline 
0.5 & 0.0281 & 0.0259 & 0.0722 & 0.0689 \\\hline 
0.7 & 0.0438 & 0.0413 & 0.1359 & 0.1340 \\\hline
0.9 & 0.1486 & 0.1481 & 0.4510 & 0.4213 \\\hline
\end{tabular}
\caption{$f=1$}
\label{table:varDiffFrac1}
\end{subfigure}
\\
\begin{subfigure}[b]{0.44\textwidth}
\begin{tabular}{|c|c|c|c|c|}
\hline
\multirow{2}{*}{$\alpha$} & \multicolumn{2}{|c|}{uniform} & \multicolumn{2}{|c|}{pareto} \\ \cline{2-5}
{} & simulations & analysis & simulations & analysis \\\hline 
0.1 & 0.0282 & 0.0299 & 0.0919 & 0.0867 \\\hline
0.3 & 0.0384 & 0.0373 & 0.1040 & 0.1176 \\\hline 
0.5 & 0.0568 & 0.0510 & 0.1902 & 0.1734 \\\hline 
0.7 & 0.0880 & 0.0853 & 0.2966 & 0.2887 \\\hline
0.9 & 0.3138 & 0.3098 & 0.9718 & 0.8346 \\\hline
\end{tabular}
\caption{$f=3$}
\label{table:varDiffFrac3}
\end{subfigure}
\caption{Simulation results for the variance of Algorithm \ref{alg:diffEstWithSampling} using uniform and Pareto distributions. 
The ``analysis'' variance is determined according to Theorem \ref{thmDiff}.}
\label{table:varDiff}
\end{table}

\subsection{The Case of $k>2$ Streams: Algorithm \ref{alg:generalEstWithSampling}}

We now seek to verify the variance analysis of Algorithm \ref{alg:generalEstWithSampling} (Theorem \ref{thmKbigT2}).
We consider three streams, each with $10^4$ unique elements and uniformly distributed frequencies as described before.
Let us denote the streams $A,B$ and $C$. In this section we are interested in estimating $\abs{(A\cap B)\setminus C}$. One can easily see this cardinality satisfies
\begin{equation*}
\abs{(A\cap B)\setminus C} = \abs{A\cap B} - \abs{A\cap B\cap C}{.}
\end{equation*}
In the simulation test below we fix the cardinality of $\abs{A\cap B\cap C}$ and estimate $\widehat{\abs{(A\cap B)\setminus C}}$ for different values of the intersection $\abs{A\cap B}$ using Algorithm \ref{alg:generalEstWithSampling}.
To sum up, we simulate the three streams with the following cardinalities:
\begin{enumerate}
\item $\abs{A}=\abs{B}=\abs{C}=10^4$.
\item $\abs{A\cap C}=\abs{B\cap C}=2,000$.
\item $\abs{A\cap B\cap C}=1,000$.
\end{enumerate}
We then estimate $\abs{(A\cap B)\setminus C}$ for different values of the intersection $\abs{A\cap B}$. 

Table~\ref{table:varThree} presents the simulation results for different intersection values for $m=100$ buckets and $u=1,000$. The sampling ratio is $P=1/100$. The results are averaged again over $1,000$ runs, 
and the ``analysis'' variance is determined according to Theorem \ref{thmKbigT2}.
We can clearly see that the variance of the algorithm as found by the simulations is very close to the variance found by our analysis (relative error of about $5 \%$). As expected, when $\abs{A\cap B}$ increases (and the estimated quantity $\abs{(A\cap B)\setminus C}$ increases as well), the variance decreases.

\begin{table}[ht]
\centering 
\begin{tabular}{|c|c|c|}
\hline
$\abs{A\cap B}$ & simulations & analysis \\ \hline
1500 & 0.5361 & 0.5102 \\\hline
2000 & 0.2345 & 0.2501 \\\hline 
3000 & 0.1139 & 0.1201 \\\hline 
4000 & 0.0689 & 0.0767 \\\hline
5000 & 0.0516 & 0.0551 \\\hline
6000 & 0.0367 & 0.0421 \\\hline
7000 & 0.0281 & 0.0334 \\\hline
8000 & 0.0243 & 0.0272 \\\hline
8500 & 0.0228 & 0.0247 \\\hline
\end{tabular}
\caption{Simulation results for Algorithms \ref{alg:generalEstWithSampling} using uniform distribution.
The ``analysis'' variance is determined according to Theorem \ref{thmKbigT2}.}
\label{table:varThree}
\end{table}


\section{Conclusions} \label{sec:conclusion}

In this paper we studied the generalization of the cardinality estimation problem to multiple streams, when only a small sample of each stream is given. 
We presented a new framework of sketch algorithms that combine sub-linear space and sampling. The new framework, called Maximal-Term with Subsample (MTS), can accurately estimate the cardinality of any set expression between any number of streams using only a small sample of each stream. We presented three algorithms that address any set expression between two streams. We then presented another algorithm that extends these algorithms to the case of $k>2$ streams.
We showed that all our algorithms are unbiased, and we analyzed their asymptotic variance. Finally, we presented simulation results that validate our analysis.

\bibliographystyle{abbrv}
\bibliography{MTS}

\appendix
\section*{Appendix}

\renewcommand\thesubsection{\Alph{subsection}}

We now analyze our proposed algorithms. To shorten the notations, we use $n$ to denote the wanted cardinality in each algorithm.

We start with a preliminary lemma that shows how to compute the probability distribution of a product of two normally distributed random variables whose covariance is $0$:
\begin{lemma}[Product distribution]\label{lemma:productDis} \ \\
Let $X$ and $Y$ be two random variables satisfying that $X \to \Normal{\mu_x}{\sigma_x^2}$, $Y \to \Normal{\mu_y}{\sigma_y^2}$, and $\Cov{X}{Y}=0$.
Then, the product $X \cdot Y$ asymptotically satisfies the following:
\begin{equation*}
X\cdot Y \to \Normal{\mu_x\mu_y}{\mu_y^2\sigma_x^2 + \mu_x^2\sigma_y^2} \text{.}
\end{equation*}
\end{lemma} 
A proof is given in \cite{Shao1998}.

\subsection{Cardinality Estimation without Sampling Using the HyperLogLog Algorithm}

For the rest of the appendix we use the HyperLogLog algorithm \cite{Flajolet2007} as \cardProc{}.
This estimator belongs to the family of max sketches and is the best known cardinality estimator. Its standard error is $ 1.04 /\sqrt{m}$, where $m$ is the number of storage units. Its pseudo-code is as follows:

\begin{algorithm2} 
{The HyperLogLog algorithm for the cardinality estimation problem}
\label{alg:hyperLogLog}
\begin{enumerate}
\item Initialize $m$ registers: $C_1,C_2,\ldots,C_m$ to 0.
\item For each input element $x_i$ do:
\begin{enumerate}
\item Let $\rho= \floor{-\log_2\paren{h_1(x_i)}}$ be the leftmost 1-bit position of the hashed value.
\item Let $j=h_2(x_i)$ be the bucket for this element.
\item $C_j \leftarrow \max{\set{C_j,\rho}}$.
\end{enumerate}
\item To estimate the value of $n$ do:
\begin{enumerate}
\item $Z \leftarrow {(\sum_{j=1}^{m}2^{-C_j})}^{-1}$ is the harmonic mean of $2^{C_j}$.
\item return $\alpha_m m^2 Z$, where \\
$\alpha_m= \paren{m \int^{\infty}_{0}\paren{\log_2{\paren{\frac{2+u}{1+u}}}}^m \,du }^{-1}$.
\end{enumerate}
\end{enumerate}
\end{algorithm2}

The following Lemma summarizes the statistical performance of Algorithm \ref{alg:hyperLogLog} \textbf{without sampling}, i.e., when the algorithm processes the entire stream:

\begin{lemma} \label{lemma:hyper} \ \\
For Algorithm \ref{alg:hyperLogLog}, 
$\widehat{c} \to \Normal{c}{\frac{c^2}{m}}$, where $c$ is the actual cardinality of the considered set, $\widehat{c}$ is the estimate computed by the algorithm, and $m$ is the number of storage units used by the algorithm.
\end{lemma}
The proof is given in \cite{Flajolet2007}.

\begin{corollary} \label{cor:union} \ \\
Let $A$ and $B$ be two streams. When Algorithm \ref{alg:hyperLogLog} is used with $m$ storage units, and without sampling, the following holds:
\begin{equation*} 
\widehat{\abs{A}} \to \Normal{\abs{A}}{\frac{\abs{A}^2}{m}} \text{,} \:\:\: \widehat{\abs{B}} \to \Normal{\abs{B}}{\frac{\abs{B}^2}{m}} \text{,}
\end{equation*}
\begin{equation*}
\text{and} \:\:\:\:\: \widehat{\abs{A \cup B}} \to \Normal{\abs{A \cup B}}{\frac{\abs{A \cup B}^2}{m}} \text{.}
\end{equation*}
\end{corollary}

\subsection{Cardinality Estimation with Sampling}

\ifthenelse{\boolean{doubleBlind}}
{In \cite{Cohen15} the authors analyze the asymptotic bias and variance of Algorithm \ref{alg:cardEstWithSampling}, assuming that the HyperLogLog algorithm \cite{Flajolet2007} is used as \cardProc{}. They prove that the sampling does not affect the asymptotic unbiasedness of the estimator and analyze the effect of the sampling rate $P$ on the estimator's variance, with respect to the storage sizes $m$ and $u$. The following theorem summarizes the statistical performance of the algorithm:}
{In a previous paper \cite{Cohen15} we analyzed the asymptotic bias and variance of Algorithm \ref{alg:cardEstWithSampling}, assuming that the HyperLogLog algorithm \cite{Flajolet2007} is used as \cardProc{}. We proved that the sampling does not affect the asymptotic unbiasedness of the estimator and analyzed the effect of the sampling rate $P$ on the estimator's variance, with respect to the storage sizes $m$ and $u$. The following Theorem summarizes the statistical performance of the algorithm:}

\begin{thm} \ \\
\label{thmProb1Samp}
Algorithm \ref{alg:cardEstWithSampling} estimates $n$ with mean value $n$ and variance $\frac{n^2}{u}\frac{P_{0,1}}{(1-P_0)^2} + \frac{n^2}{m}$, namely, $\widehat{n} \to \Normal{n}{\frac{n^2}{u}\frac{P_{0,1}}{(1-P_0)^2} + \frac{n^2}{m}}$, where $P_{0,1}=2P_0(1-P_0) + P_1$. 
In addition, $P_0$ and $P_1$ satisfy:
\begin{enumerate}
\item $\Ex{P_0} = \frac{1}{n}\sum_{i=1}^{n}{e^{-P\cdot f_i}}$.
\item $\Ex{P_1} = \frac{P}{n}\sum_{i=1}^{n}{f_i \cdot e^{-P\cdot f_i}}$
\end{enumerate}
where $\set{e_1,e_2,\ldots,e_n}$ are the distinct elements in the full (unsampled) stream, and $f_i$ is the frequency of element $e_i$ in the full stream.
\end{thm}
A proof is given in \cite{Cohen15}.

\subsection{Cardinality Estimation of Set Union with Sampling}

As described in Section \ref{sub:union}, estimating the set union cardinality using Algorithm \ref{alg:unionEstWithSampling} is equivalent to estimating the cardinality of a single stream $A\cup B$ based on its sampled stream, $S_{\cup}$. Thus, the statistical performance of Algorithm \ref{alg:unionEstWithSampling} is equal to that of Algorithm \ref{alg:cardEstWithSampling}:

\begin{corollary} \ \\
\label{corUnion}
Algorithm \ref{alg:unionEstWithSampling} estimates $n=\abs{A\cup B}$ with mean value $n$ and variance $\frac{n^2}{u}\frac{P_{0,1}^U}{(1-P_0^U)^2} + \frac{n^2}{m}$, namely, 
\begin{equation*}
\widehat{\abs{A\cup B}} \to \Normal{n}{\frac{n^2}{u}\frac{P_{0,1}^U}{(1-P_0^U)^2} + \frac{n^2}{m}} \text{,}
\end{equation*}
where $P_0^U$ and $P_{0,1}^U$ are as stated in Theorem \ref{thmProb1Samp} with respect to the union stream $A\cup B$.
\end{corollary}

\subsection{Cardinality Estimation of Set Intersection and Set Difference with Sampling}\label{appendix:interANDdiff}

We will use the following lemmas:

\begin{lemma} \ \\
\label{lemma1minusP}
$\widehat{P_0} \to \Normal{P_0}{\frac{P_{0,1}}{u}}$, and 
$\widehat{\frac{1}{1-P_0}} \to \Normal{\frac{1}{1-P_0}}{\frac{1}{u}\frac{P_{0,1}}{(1-P_0)^4}}$, where $P_{0,1}=2P_0(1-P_0) + P_1$ and $u$ is the length of the subsample stream. 
\end{lemma}
A proof is given in \cite{Cohen15}.

\begin{lemma} 
\label{lemmaSamples} \ \\
Procedure \ref{proc:da} estimates $d_A=\frac{\abs{S_A}}{\abs{S_{\cup}}}$ with mean value $d_A$ and variance $\frac{1}{f}d_A(1-d_A)$, namely, $\widehat{d_A} \to \Normal{d_A}{\frac{1}{f}d_A(1-d_A)}$, where $f = \abs{\text{MTS}_2(A\cup B)}$ is the cardinality of $\text{MTS}_2(A\cup B) = U_{A\cup B}$.
\end{lemma}
\begin{proof}\ \\
Procedure \ref{proc:da} estimates $d_A$ as $\widehat{d_A} = \frac{\abs{\text{MTS}_2(A)\cap \text{MTS}_2(A\cup B)}}{f}$.
Denote the distinct elements in the union subsample as $\text{MTS}_2(A\cup B)=\set{u_1,u_2,\ldots,u_f}$. For each $u_j$, the probability that it belongs to $S_A$ is:
\begin{equation*}\label{eq:pr}
\Prob{u_j \in S_A \:|\: u_j \in S_{\cup}} = \frac{\Prob{u_j \in S_A}}{\Prob{u_j \in S_{\cup}}} = \frac{\abs{S_A}}{\abs{S_{\cup}}}=d_A \text{.}
\end{equation*}
It follows that $\widehat{d_A}$ is a sum of $f$ Bernoulli variables with success probability $d_A$. Therefore, it is binomially distributed, and can be asymptotically approximated using normal distribution as $f \to \infty$; namely, $\widehat{d} = \frac{\abs{\text{MTS}_2(A)\cap \text{MTS}_2(A\cup B)}}{f} \to \Normal{d_A}{\frac{1}{f}d_A(1-d_A)}$.
\end{proof}

\begin{lemma} \label{lem:cov}\ \\
The covariance of $\widehat{X_A}=\widehat{\frac{1}{1-P_0^A}}\cdot\widehat{d_A}\cdot\widehat{(1-P_0^\cup)}$ and $\widehat{X_B}$ (defined similarly) satisfies $\Cov{\widehat{X_A}}{\widehat{X_B}} = \frac{(1-P_0^\cup)^2}{f(1-P_0^A)(1-P_0^B)}(\frac{\abs{S_\cap}}{\abs{S_\cup}} - \frac{\abs{S_A}\abs{S_B}}{\abs{S_\cup}^2})$, where $f = \abs{\text{MTS}_2(A\cup B)}$ is the cardinality of $\text{MTS}_2(A\cup B) = U_{A\cup B}$.
\end{lemma}
\begin{proof} \ \\
Recall that $\widehat{X_A} = \widehat{\frac{1}{1-P_0^A}}\cdot\widehat{d_A}\cdot\widehat{(1-P_0^\cup)}$, and similarly for $X_B$.
The dependence is between $\widehat{d_A}$ and $\widehat{d_B}$; thus it follows from covariance properties that
\begin{equation}\label{eq:1}
\Cov{\widehat{X_A}}{\widehat{X_B}} = \frac{(1-P_0^\cup)^2}{f(1-P_0^A)(1-P_0^B)}\cdot \Cov{\widehat{d_A}}{\widehat{d_B}}\text{.}
\end{equation}
Let us denote the distinct elements in the union subsample as $\text{MTS}_2(A \cup B) = \set{u_1,\ldots,u_f}$. Recall from Procedure \ref{proc:da} that $\widehat{d_A}$ can be written as follows:
\begin{equation}\label{eq:da}
\widehat{d_A} = \widehat{\frac{\abs{S_A}}{\abs{S_\cup}}} = \frac{\sum_{j=1}^{f}{I_j^A}}{f}\text{,}
\end{equation}
where $I_j^A$ is an indicator variable that gets $1$ if $u_j\in S_A$ and $0$ otherwise. Similarly we can write $\widehat{d_B}$ using indicator variables $I_j^B$ that get $1$ if $u_j \in S_B$ and $0$ otherwise.

Using covariance properties and Eq. (\ref{eq:da}), we obtain that
\begin{align}\label{eq:2}
\Cov{\widehat{d_A}}{\widehat{d_B}} &= \frac{1}{f^2}\sum_{j=1}^{f}\sum_{t=1}^{f}{\Cov{I_j^A}{I_t^B}} \nonumber \\
&= \frac{1}{f^2}\sum_{j=1}^{f}{\Cov{I_j^A}{I_j^B}} \nonumber \\
&= \frac{1}{f^2}\sum_{j=1}^{f}{(\Ex{I_j^A\cdot I_j^B} - \Ex{I_j^A}\cdot \Ex{I_j^B})} \nonumber \\
&= \frac{1}{f^2}\sum_{j=1}^{f}{(\frac{\abs{S_\cap}}{\abs{S_\cup}} - \frac{\abs{S_A}}{\abs{S_\cup}}\cdot \frac{\abs{S_B}}{\abs{S_\cup}})} \nonumber \\
&=\frac{1}{f}(\frac{\abs{S_\cap}}{\abs{S_\cup}} - \frac{\abs{S_A}\abs{S_B}}{\abs{S_\cup}^2}) \text{.}
\end{align}
The first and third equalities are due to covariance properties. The second equality is due to the independence of $I_j^A$ and $I_t^B$ when $j \neq t$.
The fourth equality is due to Lemma \ref{lemmaSamples}. 
Note that $\Prob{I_j^A=I_j^B=1}=\frac{\abs{S_\cap}}{\abs{S_\cup}}$ follows in the same way as the proof of Lemma \ref{lemmaSamples}. 
The last equality is due to algebraic manipulations.

The result follows by substituting Eq. (\ref{eq:2}) into Eq. (\ref{eq:1}).
\end{proof}

Finally, the following theorem states the asymptotic statistical performance of Algorithm \ref{alg:interEstWithSampling}:

\begin{thm} \ \\
\label{thmInter}
Algorithm \ref{alg:interEstWithSampling} estimates $n=\abs{A\cap B}$ with mean value $n$ and variance $V$, namely, $\widehat{\abs{A\cap B}} \to \Normal{n}{V}$. $V$ satisfies:
\begin{align*}
V &= \frac{n^2}{m} + \frac{n^2}{u}\frac{P_{0,1}^U}{(1-P_0^U)^2} -\frac{1}{f}(n+\abs{A\cup B})^2 \\
&+\frac{1}{f}(\frac{\abs{A}\abs{S_\cup}}{1-P_0^A} + \frac{\abs{B}\abs{S_\cup}}{1-P_0^B} + \frac{2\abs{S_\cap}\abs{S_\cup}}{(1-P_0^A)(1-P_0^B)}) \text{,}
\end{align*}
where $f = \abs{\text{MTS}_2(A\cup B)}$.
\end{thm}
\begin{proof} \ \\
Let us denote $\widehat{X_A} = \widehat{\frac{1}{1-P_0^A}}\cdot\widehat{d_A}\cdot\widehat{(1-P_0^\cup)}$ (and similarly for $\widehat{X_B}$). Thus, we can rewrite the estimator in Algorithm \ref{alg:interEstWithSampling} (Eq. (\ref{eq:jacEst2})) as 
\begin{equation*}
\widehat{\rho}= \widehat{X_A} + \widehat{X_B} - 1 \text{.}
\end{equation*}
We first analyze the asymptotic distribution of $\widehat{X_A}$.
Recall that according to Good-Turing it follows that:
\begin{equation}\label{eq:ez1}
\abs{A} = \frac{1}{1-P_0^A}\cdot \abs{S_A} \:\: \text{and,} \:\: \frac{\abs{S_A}}{\abs{A}} = 1-P_0^A {.}
\end{equation}
Applying Lemma \ref{lemma:productDis} on $\widehat{X_A}$, we get for the expectation
\begin{equation*}
\Ex{\widehat{X_A}} = \frac{1}{1-P_0^A} \cdot d_A \cdot (1-P_0^\cup) = \frac{\abs{A}}{\abs{A\cup B}} \text{.}
\end{equation*}
The second equality follows by substituting $d_A=\frac{\abs{S_A}}{\abs{S_\cup}}$ and using Eq. (\ref{eq:ez1}).

For the variance we get,
\begin{align*}
\Var{\widehat{X_A}} &= \Var{\widehat{\frac{1}{1-P_0^A}}\cdot\widehat{d_A}\cdot\widehat{(1-P_0^\cup)}} \\
&\to \Ex{\widehat{\frac{1}{1-P_0^A}}}^2\cdot\Ex{\widehat{(1-P_0^\cup)}}^2\cdot \Var{\widehat{d_A}} \\
&= (\frac{1-P_0^\cup}{1-P_0^A})^2\cdot\Var{d_A} \\
&= (\frac{1-P_0^\cup}{1-P_0^A})^2\cdot \frac{1}{f}(d_A(1-d_A))
\text{.}
\end{align*}
The first equality is due to the definition of $\widehat{X_A}$. 
The limit is because $\Var{\widehat{\frac{1}{1-P_0^A}}} \to 0$ and $\Var{\widehat{(1-P_0^\cup)}} \to 0$. The last equality follows Lemmas \ref{lemma1minusP} and \ref{lemmaSamples}.
In total, we get that 
\begin{equation}\label{eq:dis}
\widehat{X_A} \to \Normal{\frac{\abs{A}}{\abs{A\cup B}}}{(\frac{1-P_0^\cup}{1-P_0^A})^2\cdot \frac{1}{f}(d_A(1-d_A))} \text{.}
\end{equation} 
The same holds symmetrically for $\widehat{X_B}$. 

We now consider $\widehat{\rho} = \widehat{X_A} + \widehat{X_B} - 1$.
Note that $\widehat{X_A}$ and $\widehat{X_B}$ are dependent variables. In Lemma \ref{lem:cov} we prove that $\Cov{\widehat{X_A}}{\widehat{X_B}} = \frac{(1-P_0^\cup)^2}{f(1-P_0^A)(1-P_0^B)}(\frac{\abs{S_\cap}}{\abs{S_\cup}} - \frac{\abs{S_A}\abs{S_B}}{\abs{S_\cup}^2})$.
For the expectation we get
\begin{align*}
\Ex{\widehat{\rho}} &= \Ex{\widehat{X_A}} + \Ex{\widehat{X_B}} - 1 \\
&= \frac{\abs{A}}{\abs{A\cup B}} + \frac{\abs{B}}{\abs{A\cup B}} - 1 \\
&= \frac{\abs{A\cap B}}{\abs{A\cup B}} = \rho\text{.}
\end{align*}
It follows that $\widehat{\rho}$ is an unbiased estimator for $\rho$. For the variance we get
\begin{align}\label{eq:varBIG}
\Var{\widehat{\rho}} &= \Var{\widehat{X_A}} + \Var{\widehat{X_B}} +2\Cov{\widehat{X_A}}{\widehat{X_B}} \nonumber \\
&= \frac{(1-P_0^\cup)^2}{f}(Z_A + Z_B)  \nonumber \\
&+ \frac{2(1-P_0^\cup)^2}{f(1-P_0^A)(1-P_0^B)}(\frac{\abs{S_\cap}}{\abs{S_\cup}}-\frac{\abs{S_A}\abs{S_B}}{\abs{S_\cup}^2})\text{,}
\end{align}
where $Z_A = \frac{\abs{A}}{\abs{S_\cup}(1-P_0^A)} - (\frac{\abs{A}}{\abs{S_\cup}})^2$, and similarly for $B$. The first equality is due to variance properties and the second equality follows from Eq. (\ref{eq:dis}) and Lemma \ref{lem:cov}.

In total we obtain that $\widehat{\rho} \to \Normal{\rho}{V_1}$, where $V_1$ is as stated in Eq. (\ref{eq:varBIG}).
Applying Lemma \ref{lemma:productDis} on the independent variables $\widehat{\rho}$ and $\widehat{\abs{A \cup B}}$ concludes the proof.
\end{proof}

Similarly, for set difference we get:
\begin{thm} \ \\
\label{thmDiff}
Algorithm \ref{alg:diffEstWithSampling} estimates $n=\abs{A\setminus B}$ with mean value $n$ and variance $V$, namely, $\widehat{\abs{A\cap B}} \to \Normal{n}{V}$, where
\begin{equation*}
V = \frac{n^2}{m} + \frac{n^2}{u}\frac{P_{0,1}^U}{(1-P_0^U)^2} + \frac{1}{f}(\frac{\abs{B}\abs{S_\cup}}{1-P_0^B} - \abs{B}^2)\text{,}
\end{equation*}
where $f$ is as stated in Theorem \ref{thmInter}.
\end{thm}

\subsection{Cardinality Estimation of Set Expressions With Sampling for $k>2$ Streams}\label{appendix:bigTk}

We will use the following lemma:

\begin{lemma} \label{lemma:jacard} \ \\
In Eq. (\ref{eq:jacest}), the estimation of $\widehat{\rho}$ is normally distributed with mean $\rho$ and variance $\frac{1}{m}\rho(1-\rho)$; i.e., 
\begin{equation*}
\widehat{\rho} \to \Normal{\rho}{\frac{1}{m}\rho(1-\rho)} \text{.}
\end{equation*} 
The same holds for the estimation of $\rho_>$ and $\rho_<$ according to Eq. (\ref{eq:jacminus}), with the change of $\rho$ to $\rho_>$ and $\rho_<$ respectively.
\end{lemma}
\begin{proof} \ \\
We prove the lemma for $\rho$. The proof for $\rho_>$ and $\rho_<$ is similarly.
According to \cite{Broder97new}, for the $j$'th hash function the following holds:
\begin{equation}
\label{eq:pro}
\Prob{x_A^j=x_B^j}=\frac{\abs{A \cap B}}{\abs{A \cup B}} \text{.}
\end{equation}
The intuition is to consider the hash function $h_j$ and define $m(S)$, for every set $S$, as the element in $S$ whose hash value for $h_j$ is maximum; i.e., $h_j(m(S)) = x_S^j$. Then, $m(A)=m(B)$ holds only when $m(A \cup B)$ lies in $A \cap B$. The probability for this is the Jaccard ratio $\rho$, and therefore $\Prob{x_A^j = x_B^j} = \rho$.

From Eqs. (\ref{eq:jacest}) and (\ref{eq:pro}) follows that $\widehat{\rho}$ is a sum of $m$ Bernoulli variables. Therefore, it is binomially distributed, and can be asymptotically approximated to normal distribution as $m\to \infty$; namely, 
\begin{equation*}
\widehat{\rho} = \frac{\sum_{l=1}^{m}I_{x_A^j=x_B^j}}{m} \to \Normal{\rho}{\frac{1}{m}\rho(1-\rho)} \text{.}
\end{equation*}
\end{proof}

Finally, the following theorem states the asymptotic statistical performance of Algorithm \ref{alg:generalEstWithSampling}:

\begin{thm} \ \\
\label{thmKbigT2}
Algorithm \ref{alg:generalEstWithSampling} estimates $n=\abs{A_1 (\star 1) A_2 (\star 2) \ldots (\star (k-1)) A_k}$ with mean value $n$ and variance $\frac{n^2}{g}\cdot\frac{P_{0,1}^X}{(1-P_0^X)^2} + \frac{n \cdot \abs{S_\cup}}{m(1-P_0^X)}$, namely, $\widehat{n} \to \Normal{n}{\frac{n^2}{g}\cdot\frac{P_{0,1}^X}{(1-P_0^X)^2} + \frac{n \cdot \abs{S_\cup}}{m(1-P_0^X)}}$, \\
where $X = A_1 (\star 1) A_2 (\star 2) \ldots (\star (k-1)) A_k$ is the full (unsampled) stream and $g$ is the length of the subsample stream, as described in Procedure \ref{proc:gen}.
\end{thm}
\begin{proof}\ \\
The estimator is $\widehat{n} = \widehat{\rho_G}\cdot\widehat{\abs{S_\cup}}\cdot \widehat{\frac{1}{1-P_0^X}}$, and recall that $\rho_G = \frac{n_s}{\abs{S_\cup}}$.
\begin{enumerate}
\item According to Lemma \ref{lemma:jacard}, $\widehat{\rho_G} \to \Normal{\rho_G}{\frac{1}{m}\rho_G(1-\rho_G)}$.
\item According to Corollary \ref{cor:union}, $\widehat{\abs{S_\cup}} \to \Normal{\abs{S_\cup}}{\frac{\abs{S_\cup}^2}{m}}$.
\end{enumerate}
Consider the product $\widehat{\rho_G}\cdot\widehat{\abs{S_\cup}}$. 
Then, according to Lemma \ref{lemma:productDis} and because the variables are independent, we get that 
\begin{equation*}
\Ex{\widehat{\rho_G}\cdot\widehat{\abs{S_\cup}}}=\rho_G\cdot \abs{S_\cup} = n_s
\end{equation*}
and
\begin{align*}
\Var{\widehat{\rho_G}\cdot\widehat{\abs{S_\cup}}} &=\rho_G^2\cdot \frac{\abs{S_\cup}^2}{m} + \frac{1}{m}\rho_G(1-\rho_G)\cdot \abs{S_\cup}^2 \nonumber \\
&= \frac{n_s^2}{m} + \frac{1}{m}(\rho_G\abs{S_\cup}^2-n_s^2) \nonumber \\
&= \frac{n_s\cdot \abs{S_\cup}}{m}\text{.}
\end{align*}
In total, we get that $\widehat{\rho_G}\cdot\widehat{\abs{S_\cup}} \to \Normal{n_s}{\frac{n_s\cdot \abs{S_\cup}}{m}}$.

Denoting $T=\widehat{\rho_G}\cdot\widehat{\abs{S_\cup}}$, the estimator is $\widehat{n}=T\cdot \widehat{\frac{1}{1-P_0^X}}$. Thus, we are left with the final term in the estimator $\widehat{\frac{1}{1-P_0^X}}$.
According to Lemma \ref{lemma1minusP}, $\widehat{\frac{1}{1-P_0^X}} \to \Normal{\frac{1}{1-P_0^X}}{\frac{1}{g}\frac{P_{0,1}^X}{(1-P_0^X)^4}}$.  Then, according to Lemma \ref{lemma:productDis} and because the variables are independent, we get that
\begin{equation*}
\Ex{\widehat{n}} = \Ex{T\cdot \widehat{\frac{1}{1-P_0^X}}}=n_s\cdot \frac{1}{1-P_0^X} = n
\end{equation*}
and
\begin{align*}
\Var{\widehat{n}} &= \Var{T \widehat{\frac{1}{1-P_0^X}}} \nonumber \\
&= n_s^2\cdot (\frac{1}{g}\frac{P_{0,1}^X}{(1-P_0^X)^4}) + \frac{n_s\cdot \abs{S_\cup}}{m}\cdot \frac{1}{(1-P_0^X)^2} \nonumber \\
&= \frac{n^2}{g}\cdot\frac{P_{0,1}^X}{(1-P_0^X)^2} + \frac{n \cdot \abs{S_\cup}}{m(1-P_0^X)} \text{.}
\end{align*}
The last equality is due to $n=n_s\cdot \frac{1}{1-P_0^X}$, which follows from Good-Turing.
\end{proof}

\end{document}